\documentclass[10pt]{amsart}
\usepackage{amssymb,latexsym}
\usepackage[bookmarks=true,pdfstartview=FitH]{hyperref}
\usepackage{pdfsync}
\usepackage{pstcol,pstricks,pstricks-add,pst-node,pst-tree}
\usepackage{verbatim}

\theoremstyle{plain}
\newtheorem{theorem}{Theorem}[section]
\newtheorem{lemma}[theorem]{Lemma}
\newtheorem{proposition}[theorem]{Proposition}

\theoremstyle{definition}
\newtheorem{definition}[theorem]{Definition}
\newtheorem{example}[theorem]{Example}

\newcommand{\CC}[2]{\mathcal{C}_{c}(#1,#2)}	
\newcommand{\CS}[1]{\mathcal{C}_{c}(#1,#1)}	
\newcommand{\CV}{\downarrow}                	
\newcommand{\D}[2]{\mathcal{D}(#1,#2)}      	
\newcommand{\F}[1]{\mathcal{#1}}            	
\newcommand{\N}[1]{\overrightarrow{#1}}     	
\newcommand{\NB}[1]{\F{N}_{#1}}             	
\newcommand{\ST}{:}                     			

\begin{document}
\title{Differential Calculus on Cayley Graphs}
\author[Patten, et al.]{Daniel R. Patten}
\address{USAF Research Laboratory Information Directorate\\
    Rome NY 13441--4514}
\thanks{The first author is supported by the National Research Council.} 
\email{drpatten@syr.edu}
\author[]{Howard A. Blair}
\address{Department of Electrical Engineering and Computer Science\\
    Syracuse University\\
    Syracuse NY 13244--4100}
\email{blair@ecs.syr.edu}
\author[]{David W. Jakel}
\address{Department of Electrical Engineering and Computer Science\\
    Syracuse University\\
    Syracuse NY 13244--4100}
\email{dwjakel@juno.com}
\author[]{Robert J. Irwin}
\address{Department of Electrical Engineering and Computer Science\\
    Syracuse University\\
    Syracuse NY 13244--4100}
\email{rjirwin@syr.edu}

\date{\today}

\begin{abstract}
We conservatively extend classical elementary differential calculus to the Cartesian closed category of convergence spaces.  By specializing results about the convergence space representation of directed graphs, we use Cayley graphs to obtain a differential calculus on groups, from which we then extract a Boolean differential calculus, in which both linearity and the product rule, also called the Leibniz identity, are satisfied.
\end{abstract}
\maketitle

\section{Introduction}

In 2007, Blair, et al., \cite{MR2389716} conservatively extended (in the sense of Shoenfield \cite{MR1809685}) the concept of \emph{differential} from the spaces of classical analysis to arbitrary convergence spaces.  To define the conditions under which functions are locally differentiable, linear-like structure was extracted from subgroups of the automorphism groups of the convergence spaces.  In this paper, however, we eliminate the apparatus of automorphism groups by appropriately restricting the definition of \emph{differential}.  Importantly, this restriction is also a conservative extension of the concept of \emph{differential} to convergence spaces.

The present work originates from Patten \cite{dP14}, to which we refer the reader for terminology, notation, and basic results from the theory of convergence spaces.  For ease of reference, however, we present in Section \ref{SecGrpsConv} preliminary results about the convergence space representation of reflexive digraphs and its specialization to Cayley graphs.  In Section \ref{SecDffrntls}, we restrict the definition of \emph{differential} given in \cite{MR2389716} and verify that this restriction also correctly generalizes the differentials of classical analysis.  By identifying a given group with one of its Cayley graphs, we obtain in Section \ref{SecCayley} differential calculi on groups.  Applying these results to direct sums of $\mathbb{Z}_{2}$, we obtain a Boolean differential calculus, in which differentials satisfy both linearity and the Leibniz identity.

\section{Groups as Convergence Spaces}\label{SecGrpsConv}

\subsection{Convergence Spaces}

A \emph{filter} is a nonempty collection of nonempty sets, closed under reverse inclusion and finite intersection.  A filter is called \emph{principal} if it has a smallest set $A$, and is denoted by $[A]$.  The principal filter $[\{x\}]$ is called the \emph{point filter at $x$}; we omit the braces and denote it by $[x]$.  A maximal filter is referred to as an \emph{ultrafilter}.  Equivalently, a filter $\F{U}$ on $X$ is an ultrafilter if and only if for each $A \subseteq X$, exactly one of $A$ or $X - A$ belongs to $\F{U}$.  Every filter is included in some ultrafilter.  The family of filters on $X$ is denoted by $\Phi(X)$.

A \emph{convergence structure} on a set $X$ is a relation $\CV$ between the filters on $X$ and the points of $X$ such that for each point $p$, the set of all filters related to $p$ is a filter on $\Phi(X)$ that contains $[p]$.  If a filter $\F{F}$ is related to a point $p$, we write $\F{F} \CV p$ and say that $\F{F}$ \emph{converges} to $p$.  A \emph{convergence space} is a set equipped with a convergence structure.

The category of convergence spaces \textbf{CONV} includes the category of topological spaces \textbf{TOP} as a full subcategory.  Each property of topological spaces, therefore, is an instantiation of a property of convergence spaces.  For example, a function between convergence spaces is \emph{continuous} if and only if it preserves filter convergence; if the convergence spaces are topological, then this notion of continuity is equivalent to the familiar one, that is, preservation of open sets under inverse images.

\textbf{CONV}, unlike \textbf{TOP}, is Cartesian closed: the salient distinction here is that \textbf{CONV} is closed under exponents.  In particular, there is a canonical convergence structure on the set of all continuous functions between two convergence spaces.  We exploit the Cartesian closedness of \textbf{CONV} throughout this paper.

See Beattie and Butzmann \cite{MR2327514}, Binz \cite{MR0461418}, or \cite{dP14} for further discussion on the theory of convergence spaces.

\subsection{Reflexive Digraphs}\label{SecDigraphs}

A \emph{reflexive digraph} is a directed graph with a reflexive edge set.  The edge set of a reflexive digraph induces a convergence structure on the vertex set of that reflexive digraph.  This observation provides a basis for extending concepts such as continuity and differentiability from continuous to discrete structures.  In particular, the reflexive closure of a Cayley graph is a reflexive digraph.  Thus, in analogy to choosing a basis for a vector space, we identify a group with one of its non-redundant Cayley graphs, and thereby extend classical analysis to groups.  From this theory we then extract a Boolean differential calculus.

\begin{definition}\label{DefDigraphs}
Let $(V,E)$ be a reflexive digraph.  For each $v \in V$, the \emph{graph neighborhood} of $v$ is the set $\N{v} = \{ u \in V \ST (v,u) \in E \}$.  The \emph{reflexive digraph convergence structure} on $V$ is defined by
\begin{equation*}
        \F{F} \CV v \text{ if and only if } \N{v} \in \F{F}.
\end{equation*}
When no reasonable confusion is likely, we refer to a reflexive digraph $(V,E)$ by $V$.  Unless otherwise noted, we assume that all reflexive digraphs have the reflexive digraph convergence structure.
\end{definition}

We can induce a convergence structure on the vertex set of a reflexive digraph in other ways.  For example, in \cite{MR2389716}, a filter $\F{F}$ converges to a vertex $v$ of a reflexive digraph if and only if $\F{F} = [u]$ for some vertex $u \in \N{v}$; this structure, however, requires the removal of the finite intersection property of convergence structures.

\begin{example}\label{ExmPentacle}
Let $P$ denote the pentacle, as shown in Figure \ref{FigPentacle}, with the reflexive digraph convergence structure.  (In the interest of visual clarity, we draw the reflexive reduction of a reflexive digraph rather than the reflexive digraph itself.)
\begin{figure}
\begin{center}
\begin{pspicture}(-1,-1)(4,4)
    \psdots(0,1.902113)
    \psdots(1.618034,3.077684)
    \psdots(3.236068,1.902113)
    \psdots(2.618034,0)
    \psdots(0.618034,0)
    \psline(0,1.902113)(1.618034,3.077684)
    \psline(1.618034,3.077684)(3.236068,1.902113)
    \psline(3.236068,1.902113)(2.618034,0)
    \psline(2.618034,0)(0.618034,0)
    \psline(0.618034,0)(0,1.902113)
    \psset{ArrowInside=->}
    \psline(0,1.902113)(3.236068,1.902113)
    \psline(3.236068,1.902113)(0.618034,0)
    \psline(0.618034,0)(1.618034,3.077684)
    \psline(1.618034,3.077684)(2.618034,0)
    \psline(2.618034,0)(0,1.902113)
    \uput[180](0,1.902113){4}
    \uput[90](1.618034,3.077684){0}
    \uput[0](3.236068,1.902113){1}
    \uput[270](2.618034,0){2}
    \uput[270](0.618034,0){3}
\end{pspicture}
\end{center}
\caption{The Pentacle.}\label{FigPentacle}
\end{figure}
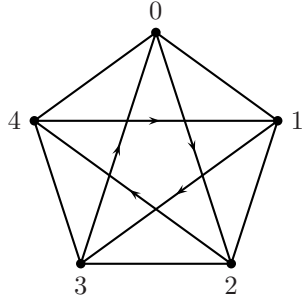
A filter $\F{F} \CV p$ if and only if $\F{F} = [A]$ for some $A \subseteq P$ that does not contain $p + 3 \pmod{5}$.  The pentacle, it should be noted, is a homogeneous pretopological space.  Each pretopological space embeds into some Cartesian product of $P$.
\end{example}

\begin{proposition}\label{PrpDigCntPt}
Let $(V_{1},E_{1})$ and $(V_{2},E_{2})$ be reflexive digraphs, and let $V_{1}$ and $V_{2}$ be the induced convergence spaces, respectively.  A function $f : V_{1} \rightarrow V_{2}$ is continuous at $v \in V_{1}$ if and only if $f(\N{v}) \subseteq \N{f(v)}$.
\end{proposition}

Propostion \ref{PrpDigCntPt} states that continuous functions between reflexive digraphs preserve graph neighborhoods; in other words, a function between reflexive digraphs is continuous if and only if it is a graph homomorphism.  Thus, the concepts of graph homomorphism and continuous function are manifestations of the same concept.

Not every convergence space is represented by a reflexive digraph.  For example, the standard topology on $\mathbb{R}$ is not a reflexive digraph.  To the contrary, if the standard topology on $\mathbb{R}$ were a reflexive digraph, then $[\N{x}]$ would be identical to the filter $\NB{x}$, the collection of all topological neighborhoods of $x$.  But the former is principal whereas the latter is non-principal.  Therefore, the standard topology on $\mathbb{R}$ is not a reflexive digraph.  In fact, no non-discrete reflexive digraph on $\mathbb{R}$ is finer than the standard topology on $\mathbb{R}$.

Every finitely generated pretopological space, however, can be represented by a reflexive digraph: for each point $p$ of a finitely generated pretopological space, define $\N{p}$ to be the set that generates the neighborhood filter $\NB{p}$ of $p$.  Subsequently, we treat all finitely generated pretopological spaces, which include all finite convergence spaces, as reflexive digraphs.  

Lowen-Colebunders and Sonck \cite{eL93} first showed that reflexive digraphs---in the present nomenclature---are exponential in \textbf{PTOP}, the category of pretopological spaces.  Applying a theorem of Nel \cite{lN77} to this result, they concluded that \textbf{REDI}, the category of reflexive digraphs, is Cartesian closed.  In particular, we have:

\begin{proposition}\label{PrpCCnvReDi}
If $X$ and $Y$ are reflexive digraphs and $f, g \in \CC{X}{Y}$, then $f \in \N{g}$ if and only if $f(a) \in \N{g(b)}$ whenever $a \in \N{b}$ for each $a$ and $b$ in $X$.
\end{proposition}

Although all reflexive digraphs are pretopological, only the transitive reflexive digraphs are topological.  Likewise, all $T_{1}$, and hence $T_{2}$, reflexive digraphs are discrete.  On the other hand, many reflexive digraphs are $T_{0}$: for a reflexive digraph to be $T_{0}$, it is both necessary and sufficient that no two of its graph neighborhoods are identical.

\subsection{Cayley Graphs}

\begin{definition}\label{DefCylyGrph}
Let $\Gamma$ be a subset of a group $G$ with identity element $e$ such that each element of $G$ is a product of elements of $\Gamma$ and no element of $\Gamma$ is a product of other elements of $\Gamma$ (that is, $\Gamma$ is \emph{non-redundant}).  We call $\Gamma$ a \emph{generating set} for $G$ and each element of $\Gamma$ a \emph{generator} of $G$.  The \emph{Cayley graph for $G$ generated by $\Gamma$} is the reflexive digraph $C$ such that the vertex set of $C$ is $G$ and the edge set of $C$ is $\{ (g,h) \ST g\gamma = h \text{ and } (\gamma = e \text{ or } \gamma \in \Gamma)\}$.
\end{definition}

\begin{example}\label{ExmCayleyS3}
Consider the Cayley graph of the symmetric group $S_{3}$, as represented in Figure \ref{FigCayleyS3}.
\begin{figure}
\begin{center}
\begin{pspicture}(-1,-1)(7,6)
    \psdots(0,0)(2.140933,1.236068)(3.140933,2.968119)
        (3.140933,5.440255)(4.140933,1.236068)(6.281865,0)
    \psset{arrows=-}
    \psline(0,0)(2.140933,1.236068)
    \psline(3.140933,2.968119)(3.140933,5.440255)
    \psline(4.140933,1.236068)(6.281865,0)
    \psset{ArrowInside=->}
    \psline(3.140933,2.968119)(4.140933,1.236068)
    \psline(4.140933,1.236068)(2.140933,1.236068)
    \psline(2.140933,1.236068)(3.140933,2.968119)
    \psline(3.140933,5.440255)(0,0)
    \psline(0,0)(6.281865,0)
    \psline(6.281865,0)(3.140933,5.440255)
    \uput[180](0,0){$tr$}
    \uput[135](2.140933,1.236068){$r^{2}$}
    \uput[180](3.140933,2.968119){$e$}
    \uput[90](3.140933,5.440255){$t$}
    \uput[45](4.140933,1.236068){$r$}
    \uput[0](6.281865,0){$tr^{2}$}
\end{pspicture}
\end{center}
\caption{A Cayley Graph for the Symmetric Group $S_{3}$.}\label{FigCayleyS3}
\end{figure}
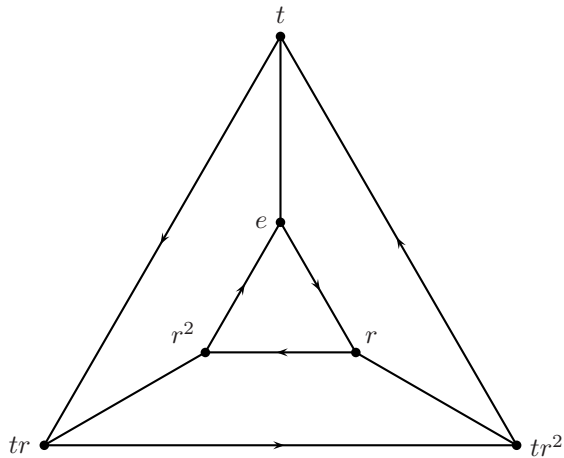
The \emph{graph} automorphisms of $S_{3}$ are precisely the left-multiplications of $S_{3}$, that is, functions of the form $\lambda x.vx$ for some $v \in S_{3}$.  The symmetric group $S_{3}$ is the automorphism group of $S^{3}$, the tertiary Cartesian product of the Sierpi\'{n}ski space .  In this role, however, the space $S_{3}$ is not the reflexive digraph of Figure \ref{FigCayleyS3}, but rather a 6-point discrete space.
\end{example}

Under Definition \ref{DefCylyGrph}, every Cayley graph is a reflexive digraph.  The additional structure of Cayley graphs, however, yields additional properties.  For example, the action of multiplication on the left by a fixed element in a group is an automorphism of the Cayley graph of that group.  With the exception of the Cayley graph for the cyclic group $\mathbb{Z}_{2}$, all nontrivial Cayley graphs are non-topological $T_{0}$ spaces; no nontrivial Cayley graph is $T_{1}$.

See \cite{dP14} for further discussion on the representation of reflexive digraphs and groups by convergence spaces.

\section{Abstract Differentials}\label{SecDffrntls}

In \cite{MR2389716}, the concept of \emph{differential} was extended from the Euclidean spaces of classical analysis to arbitrary convergence spaces.  In \cite{dP14}, we restricted the definition of differential given in \cite{MR2389716} and verified that this restriction conservatively extends elementary differential calculus, that is, it does not change the notion of differential in the context of elementary differential calculus.

\begin{definition}\label{DefDffrntls}
Let $X$ and $Y$ be convergence spaces, let $\D{X}{Y}$ be a subspace of $\CC{X}{Y}$, let $L \in \D{X}{Y}$, let $f:X \rightarrow Y$ be a function, and let $a \in X$.  Then $L$ is a \emph{differential} of $f$ at $a$ if and only if 
\begin{equation*}
	(\forall \F{A} \CV a \text{ in } X)(\exists \F{L} \CV L \text{ in } \D{X}{Y})
	(\forall K \in \F{L})(\exists A \in \F{A})
	(\forall x \in A)(f(x) \in K \cdot \{x\}).
\end{equation*}
\end{definition}

To distinguish between the differentials of elementary differential calculus and those of Definition \ref{DefDffrntls}, we call the former \emph{classical differentials}; the latter, \emph{abstract differentials}.  Theorem \ref{ThmDiffReal} establishes that on the Euclidean line these two concepts are equivalent.

The differentials are not fully determined by $X$ and $Y$---we emphasize that for all  (not necessarily distinct) convergence spaces we \emph{choose} $\D{X}{Y}$ to be a subspace of $\CC{X}{Y}$ so that  the composition of differentials is always a differential.

In the case of functions on the Euclidean line, the usual choice for $\D{\mathbb{R}}{\mathbb{R}}$ is the space of linear functions on $\mathbb{R}$.  Linearity, however, is not a necessary property of classical differentials; rather, the salient property of the space of linear functions on $\mathbb{R}$ is that it is homeomorphic to $\mathbb{R}$.  In this paper, we choose the space of affine functions for $\D{\mathbb{R}}{\mathbb{R}}$.  Of course, the differentials of classical analysis are linear, not affine; our notion of differentials produces an affine approximation of a function at a point.  This, however, is a reasonable trade-off: while we sustain only a marginal loss of simplicity in classical differential calculus, we gain a notion of differential applicable to arbitrary convergence spaces, not merely those convergence spaces with the required linear structure.  Proposition \ref{PrpAffnHome} addresses this issue for affine functions; Proposition \ref{PrpDiffEqls} establishes that a differential of a function into a $T_{1}$ space must be identical to the function at the point of differentiability (which shows that choosing $\D{\mathbb{R}}{\mathbb{R}}$ to be the space of linear functions on $\mathbb{R}$ would restrict severely the class of functions differentiable at a point in $\mathbb{R}$).

\begin{proposition}\label{PrpAffnHome}
Let $f$ be a function on $\mathbb{R}$ and let $a \in \mathbb{R}$.  Define $\mathbb{A}_{f,a}$ to be the subspace of $\CC{\mathbb{R}}{\mathbb{R}}$ consisting of all affine functions on $\mathbb{R}$ that equal $f(a)$ at $a$.  If $\mathbb{R}$ is given the standard topology, then $\mathbb{A}_{f,a}$ is homeomorphic to $\mathbb{R}$.
\end{proposition}

\begin{proof}
Define $\phi:\mathbb{R} \rightarrow \mathbb{A}_{f,a}$ by $\phi(p) = \lambda x.((x - a)p + f(a))$ for every $p \in \mathbb{R}$.  It is clear that $\phi$ is a bijection.  Since $\mathbb{R}$ is locally compact, by Corollary 1.5.17 of \cite{MR2327514}, the convergence structure of $\CS{\mathbb{R}}$ is topological and is induced by the compact-open topology.  We adapt a metric discussed in Willard \cite{MR2048350} to obtain the metric
\begin{equation*}
    d(f,g)
        = \sum_{n = 1}^{\infty} \frac{1}{2^{n}} \cdot \frac{d_{n}(f,g)}{1 + d_{n}(f,g)},
\end{equation*}
in which $d_{n}(f,g) = \sup \{|f(x) - g(x)| \ST x = [-n,n] \}$, from which it follows that $\phi$ is a homeomorphism.
\end{proof}

\begin{proposition}\label{PrpDiffEqls}
Let $X$ be a convergence space, let $Y$ be a $T_{1}$ convergence space, let $\D{X}{Y}$ be a subspace of $\CC{X}{Y}$, let $L \in \D{X}{Y}$, let $f:X \rightarrow Y$ be a function, and let $a \in X$.  If $L$ is a differential of $f$ at $a$, then $f(a) = L(a)$.
\end{proposition}

To establish that we have constructed a conservative extension of elementary differential calculus---in other words, \emph{we have not changed the notion of differential in the context of elementary differential calculus}---we first note a simplification of Definition \ref{DefDffrntls} in the case that both $X$ and $\D{X}{Y}$ are pretopological.

\begin{lemma}\label{LmaDifPTopE}
If $X$ is a pretopological space, $Y$ is a convergence space, $\D{X}{Y}$ is a pretopological subspace of $\CC{X}{Y}$, $L \in \D{X}{Y}$, $f:X \rightarrow Y$ is a function, and $a \in X$, then $L$ is a differential of $f$ at $a$ if and only if for every neighborhood $V$ of $L$, there exists a neighborhood $U$ of $a$ such that $f(x) \in V \cdot \{x\}$ for every element $x$ in $U$.
\end{lemma}

In view of Proposition \ref{PrpAffnHome}, we see that $\mathbb{A}_{f,a}$, the space of all affine functions on $\mathbb{R}$ that equal $f(a)$ at $a$, is pretopological; thus Lemma \ref{LmaDifPTopE} applies to the case of functions on $\mathbb{R}$.  We therefore obtain:

\begin{theorem}\label{ThmDiffReal}
If $\mathbb{R}$ has the standard topology, $L \in \mathbb{A}_{f,a}$, $f$ is a function on $\mathbb{R}$, and $a \in \mathbb{R}$, then $L$ is a abstract differential of $f$ at $a$ if and only if $L$ is a classical differential of $f$ at $a$.
\end{theorem}

\begin{proof}
\ [$\Rightarrow$]  Let $\varepsilon > 0$.  Let $N_{\varepsilon}$ is the $\varepsilon$-neighborhood of $m$.  Using the homeomorphism $\phi$ given by Proposition \ref{PrpAffnHome}, let $L = \phi(m)$ and $V = \phi(N_{\varepsilon})$.  By hypothesis, there exists $\delta > 0$ such that for every $x$ within $\delta$ of $a$ there is an affine function $k(h) = (h - a)m' + f(a)$ in $V$ such that $k(x) = f(x)$.  Since $k \in V$, it follows that $m' \in N_{\varepsilon}$.  Because $k(x) = f(x)$, it follows that $m' = (f(x) - f(a))/(x - a)$.  Thus
\begin{equation*}
    |(f(x) - f(a))/(x - a) - m|
        = |m' - m|
        < \varepsilon
\end{equation*}
whenever $|x - a| < \delta$, which implies that $m = \displaystyle{\lim_{x \rightarrow a} \frac{f(x) - f(a)}{x - a} = f'(a)}$, and so $L$ is a classical differential of $f$ at $a$.

[$\Leftarrow$]  Let $V$ be a neighborhood of $L(h) = (h - a)m + f(a)$.  There exists a sufficiently small $\varepsilon > 0$ such that the $\varepsilon$-neighborhood of $m$ is a subset of the $\phi$-image of $V$, where $\phi$ is the homeomorphism given by Proposition \ref{PrpAffnHome}.  By hypothesis, there exists $\delta > 0$ such that $|(f(x) - f(a))/(x - a) - m| < \varepsilon$ if $|x - a| < \delta$.  Again by Proposition \ref{PrpAffnHome}, the function $k(h) = (h - a)(f(x) - f(a))/(x - a) + f(a)$ belongs to $V$.  Since $k(x) = f(x)$, it follows that $L$ is a abstract differential of $f$ at $a$.
\end{proof}

A corollary to Proposition \ref{PrpAffnHome} is that the space of linear functions on $\mathbb{R}$ is homeomorphic to $\mathbb{R}$, and thus to $\mathbb{A}_{f,a}$ for each function $f$ on $\mathbb{R}$ and for each real number $a$.  Thus, classical differentials are, by homeomorphism, the differentials of elementary differential calculus. 

\begin{theorem}[Chain Rule]\label{ThmChainRul}
Let $X$, $Y$, and $Z$ be convergence spaces, let $a \in X$, let $f:Y \rightarrow Z$ be a function, and let $g: X \rightarrow Y$ be a function continuous at $a$.  If $L_{f}$ is a differential of $f$ at $g(a)$ and $L_{g}$ is a differential of $g$ at $a$, then $L_{f} \circ L_{g}$ is a differential of $f \circ g$ at $a$.
\end{theorem}

\begin{proof}
Let $\F{A} \CV a$.  Since $L_{g}$ is a differential of $g$ at $a$, there exists a filter $\F{K}$ converging to $L_{g}$ such that for every element $K$ in $\F{K}$, there exists a set $A$ in $\F{A}$ such that for every element $x$ in $A$, there exists a function $k$ in $K$ such that $k(x) = g(x)$.  Since $g$ is continuous at $a$, it follows that $g(\F{A}) \CV g(a)$; thus, by the hypothesis that $L_{f}$ is a differential of $f$ at $g(a)$, there exists a filter $\F{H}$ converging to $L_{g}$ such that for every element $H$ in $\F{H}$, there exists a set $B$ in $g(\F{A})$ such that for every element $x$ in $B$, there exists a function $h$ in $H$ such that $h(x) = f(x)$.

Now define the filter $\F{L} = [\{h \circ k \ST h \in H \text{ and } k \in K\} \ST H \in \F{H} \text{ and } K \in \F{K}]$.  If $\F{F} \CV x$, then $\F{H} \cdot (\F{K} \cdot \F{F}) \CV (L_{f} \circ L_{g})(x)$.  Since $\F{L} \cdot \F{F}$ includes $\F{H} \cdot (\F{K} \cdot \F{F})$, it follows that $\F{L} \CV L_{f} \circ L_{g}$.  If $L \in \F{L}$, then $L \supseteq \{h \circ k \ST h \in H \text{ and } k \in K\}$ for some $H \in \F{H}$ and $K \in \F{K}$.  Thus, there exists a set $A$ in $\F{A}$ such that for every element $x$ in $A$, there exists a function $k$ in $K$ such that $k(x) = g(x)$; likewise, there exists a set $B$ in $g(\F{A})$ such that for every element $x$ in $B$, there exists a function $h$ in $H$ such that $h(x) = f(x)$.  If $x \in A \cap g^{-1}(B)$, then there exist $k \in K$ and $h \in H$ such that $(h \circ k)(x) = h(k(x)) = h(g(x)) = f(g(x)) = (f \circ g)(x)$.  Therefore, we conclude that $L_{f} \circ L_{g}$ is a differential of $f \circ g$ at $a$.
\end{proof}

Theorem \ref{ThmChainRul} is critical to the theory of differential calculus on convergence spaces.  Suppose that a convergence space $D$ results from discretizing, by some method, a continuous structure, for example, the Euclidean line $\mathbb{R}$.  Moreover, suppose that this discretization is also continuous, in the sense that it \emph{coarsens} $\mathbb{R}$, that is, the identity function $\iota:\mathbb{R} \rightarrow D$ is continuous.  By Theorem \ref{ThmChainRul}, it follows that if $L_{f}$ is a differential of $f:D \rightarrow \mathbb{R}$ at $a$ and $L_{\iota}$ is a differential of $\iota$ at $a$, then $L_{f} \circ L_{i}$ is a differential of $f \circ \iota:\mathbb{R} \rightarrow \mathbb{R}$ at $a$.  In view of the requirement that $\D{X}{Y}$ must be a subspace of $\CC{X}{Y}$ so that the composition of differentials is always a differential, we see that differential calculus for continuous-valued functions on a discrete structure seamlessly coincides with elementary differential calculus.

\section{Differential Calculus on Groups}\label{SecCayley}

Let $C$ be a Cayley graph for a group $G$ generated by $\Gamma$; let $D$ be a Cayley graph for a group $H$ generated by $\Delta$.  The differentials of classical analysis are continuous linear maps:  they preserve both the topological structure and vector space structure of Euclidean spaces.  Thus, we choose $\D{C}{D}$ to be that subspace of $\CC{C}{D}$ the members of which preserve both the graph structure of $C$ and the group structure of $G$, that is, the continuous group homomorphisms.

Lemmas \ref{LmaCntGrpHm} and \ref{LmaCnvGrpHm} together specify the convergence structure of the subspace of $\CC{C}{D}$ the members of which are also group homomorphisms.  Lemma \ref{LmaDiffFint} is a special case of Lemma \ref{LmaDifPTopE}, in which both the domain and codomain are reflexive digraphs.  Theorem \ref{ThmDiffGrps}, which states equivalent criteria for differentiability of functions between groups, requires the introduction of some terminology.

\begin{lemma}\label{LmaCntGrpHm}
Let $C$ be a Cayley graph for a group $G$ generated by $\Gamma$; let $D$ be a Cayley graph for a group $H$ generated by $\Delta$.  A group homomorphism $\phi:G \rightarrow H$ is continuous if and only if $\phi(\N{e_{G}}) \subseteq \N{e_{H}}$ (that is, $\phi$ is continuous at $e_{G}$).
\end{lemma}

\begin{lemma}\label{LmaCnvGrpHm}
Let $C$ be a Cayley graph for a group $G$ generated by $\Gamma$; let $D$ be a Cayley graph for a group $H$ generated by $\Delta$.  If $\phi$ and $\psi$ are distinct elements of $\D{C}{D}$, then $\phi \in \N{\psi}$ if and only if there exists a unique $\delta \in \Delta$ of order 2 such that $\phi(C) \cup \psi(C) \subseteq \{ e_{H}, \delta \}$.
\end{lemma}

\begin{proof}
\ [$\Rightarrow$]  Since $\phi \neq \psi$, there exists $\gamma \in \Gamma$ such that $\phi(\gamma) \neq \psi(\gamma)$.  By hypothesis $\phi(\gamma) \in \N{\psi(\gamma)}$; thus there exists $\delta \in \Delta$ such that $\phi(\gamma) = \psi(\gamma)\delta$.  Since $\Delta$ is non-redundant, either $\phi(\gamma) = e_{H}$ and $\psi(\gamma) = \delta$ or $\phi(\gamma) = \delta$ and $\psi(\gamma) = e_{H}$.  The former case implies that $\delta^{2} = \psi(\gamma)\delta = \phi(\gamma) = e_{H}$; likewise, the latter case implies that $\delta^{2} = \phi(\gamma^{2}) \in \N{\psi(\gamma)} = \N{e_{H}}$, and so $\delta^{2} = e_{H}$.  

Now consider any other $\gamma' \in \Gamma$.  Since $\gamma\gamma' \in \N{\gamma}$, it follows that $\phi(\gamma\gamma') \in \N{\psi(\gamma)}$.  If $\phi(\gamma) = e_{H}$, then $\phi(\gamma') \in \N{\delta}$, which implies that $\phi(\gamma')$ is either $e_{H}$ or $\delta$; likewise, if $\phi(\gamma) = \delta$, then $\delta\phi(\gamma') \in \N{e_{H}}$, which also implies that $\phi(\gamma')$ is either $e_{H}$ or $\delta$.  This establishes not only that $\delta$ is unique but also that $\{e_{H},\delta\}$ includes the images of $\phi$ and $\psi$.

[$\Leftarrow$]  The desired conclusion follows immediately from Proposition \ref{PrpCCnvReDi}
\end{proof}

\begin{lemma}\label{LmaDiffFint}
If $X$ and $Y$ are reflexive digraphs, $L \in \D{X}{Y}$, $f: X \rightarrow Y$ is a function, and $a \in X$, then $L$ is a differential of $f$ at $a$ if and only if for every $x \in \N{a}$, there exists $k \in \N{L}$ such that $k(x) = f(x)$.
\end{lemma}

\begin{theorem}\label{ThmDiffGrps}
If $C$ is a Cayley graph for a group $G$ generated by $\Gamma$, $D$ is Cayley graph for a group $H$ generated by $\Delta$, $L \in \D{C}{D}$, $f: C \rightarrow D$ is a function, and $a \in C$, then $L$ is a differential of $f$ at $a$ if and only if
\begin{enumerate}
	\item	$L$ is isolated and $f(a\gamma) = L(a)L(\gamma)$ for every $\gamma \in \N{e_{G}}$;
	\item	$L$ is not isolated but constant, $f(\N{a}) \subseteq \{ e_{H} \} \cup \{ \delta \in \Delta : \delta^{2} = e_{H} \}$, and $f(e_{G}) = e_{H}$ if $e_{G} \in \N{a}$; or
	\item	$L$ is not isolated and not constant, there exists a unique $\delta \in \Delta$ of order 2 such that $f(\N{a}) \subseteq L(C) = \{e_{H},\delta\}$, and for every $\gamma \in \N{e_{G}}$, if $a\gamma$ is $e_{G}$ or has odd order, then $f(a\gamma) = e_{H}$.
\end{enumerate}
\end{theorem}

Since the edge relation on $\D{C}{D}$ is symmetric, it follows that a differential of $f:C \rightarrow D$ at $a$ need not be unique if it is neither isolated nor constant.  Uniqueness of differentials, albeit not required by Definition \ref{DefDffrntls}, is certainly a desideratum.  These non-unique differentials, however, contain the same information about the behavior of a function on the graph neighborhood at the point of differentiation: no change within a certain tolerance (in other words, movement on a two-cycle).  Unless the generating set for $D$ has an element of order 2 and $C$ has an element of even order if $C$ is finite, then $\D{C}{D}$ will contain only isolated points, that is, it will be a discrete space.  If a differential of $f:C \rightarrow D$ at $a$ is isolated, it is indeed unique. 

A classical differential of a function at a point is an affine approximation to the function that is identical to the function at that point.  Although this approximation can be made arbitrarily close to the actual function on a sufficiently small neighborhood, it is not usually identical to the function (unless the function itself is linear on some neighborhood of the point).  Since every point of a reflexive digraph, however, does have a smallest neighborhood---namely its graph neighborhood---we expect that a differential of a function on a reflexive digraph at a point to coincide with that function on the graph neighborhood of that point.  Lemma \ref{LmaDiffFint}, without any \emph{a priori} restrictions on the choice of differentials, guarantees only that a differential of a function at a point is a ``fuzzy'' approximation to that function on the graph neighborhood of that point; Theorem \ref{ThmDiffGrps}, due to the constraint that differentials are group homomorphisms, requires that the isolated differential of a function at a point is identical to that function on the graph neighborhood of that point.

If $L$ is a differential of $f$ at $a$ and $L$ is either isolated or not constant, then $f$ is continuous at $a$; otherwise, it is possible for $f$ to be discontinuous at $a$---see Example \ref{ExmBadBlean}.

\begin{example}\label{ExmIntegers}
The results of this paper enable us to obtain a differential calculus on the additive group of integers.  Identify $\mathbb{Z}$ with the Cayley graph for $\mathbb{Z}$ generated by $\{1\}$, as shown in Figure \ref{FigIntegers}.
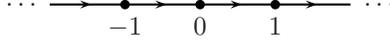
\begin{figure}
\begin{center}
\begin{pspicture}(0,0)(4,1)
    \psdots(1,0)(2,0)(3,0)
    \psset{ArrowInside=->}
    \psline(0,0)(1,0)
    \psline(1,0)(2,0)
    \psline(2,0)(3,0)
    \psline(3,0)(4,0)
    \uput[180](0,0){$\ldots$}
    \uput[270](1,0){$-1$}
    \uput[270](2,0){$0$}
    \uput[270](3,0){$1$}
    \uput[0](4,0){$\ldots$}
\end{pspicture}
\end{center}
\caption{A Cayley Graph of $\mathbb{Z}$.}\label{FigIntegers}
\end{figure}
Lemma \ref{LmaCnvGrpHm} implies that $\D{\mathbb{Z}}{\mathbb{Z}}$ is discrete and contains only the constant function $\lambda n.0$ and the identity function.  From Theorem \ref{ThmDiffGrps} it follows that $f:\mathbb{Z} \rightarrow \mathbb{Z}$ is differentiable at $n$ if and only if $f(n) = f(n + 1) = 0$ or $f(n) = n$ and $f(n + 1) = n + 1$.
\end{example}

\begin{example}\label{ExmIntPlane}
Let $C$ be a Cayley graph for a group $G$ generated by $\Gamma$.  A Cayley graph for $G \oplus G$ is the graph Cartesian product $C \times C$ generated by $(\Gamma \times \{e\}) \cup (\{e\} \times \Gamma)$.  Define the function $\ast : C \times C \rightarrow C$ by $\ast(a,b) = ab$ for all elements $a$ and $b$ of $C$.  If $G$ is abelian, then $\ast$ is continuous.  In fact, if $C$ is the Cayley graph of Figure \ref{FigIntegers}, then $\ast$ is differentiable everywhere since it belongs to $\D{\mathbb{Z}^{2}}{\mathbb{Z}}$.  The other three elements of $\D{\mathbb{Z}^{2}}{\mathbb{Z}}$ are the constant function $\lambda n.0$ and the two projections.
\end{example}

\begin{example}\label{ExmDiagonal}
Consider the diagonal function on a non-trivial Cayley graph $C$, that is, the function $d : C \rightarrow C \times C$ defined by $d(a) = (a,a)$ for each $a \in C$.  If $a \in C$, then $d(\N{a}) \cap \N{d(a)} \subseteq \{(a,a)\}$, which implies that $d(\N{a}) \not \subseteq \N{d(a)}$.  Thus $d$ is nowhere continuous.  Nor is $d$ differentiable anywhere.  To the contrary, if $d$ were differentiable at $a$, then $d(\N{a}) \subseteq \{(e,e)\} \cup \{(\gamma,e),(e,\gamma) \ST \gamma \in \Gamma \text{ and } \gamma^{2} = e \}$, from which it follows that $a = e$.  But then $\gamma \in \Gamma$ implies that $(\gamma,\gamma) \in \N{(e,e)}$, which is absurd.
\end{example}

\section{Application: Boolean Differential Calculus}\label{SecBoolean}

We now consider, as an application of Section \ref{SecCayley}, a Boolean differential calculus.

\begin{definition}\label{DefBoolCube}
For every positive integer $n$, the \emph{$n$-Boolean hypercube}, denoted by $B^{n}$, is the Cayley graph for the group $\bigoplus_{i = 1}^{n} \mathbb{Z}_{2}$, in which $\{1\}$ is the generator set of $\mathbb{Z}_{2}$.  
\end{definition}

In other words, the $n$-Boolean hypercube is the Cartesian product $\prod_{i = 1}^{n} \{0,1\}$ equipped with the reflexive digraph convergence structure defined by the condition: $b \in \N{a}$ if and only if there exists at most one $1 \leq i \leq n$ such that $\pi_{i}(b) \neq \pi_{i}(a)$.  For example, Figure \ref{FigBoolCube} depicts the 3-Boolean hypercube.

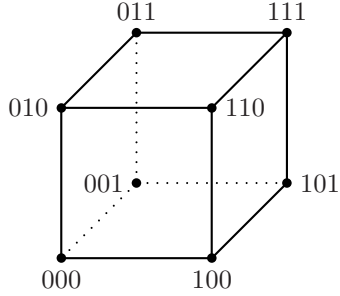
\begin{figure}
\begin{center}
\begin{pspicture}(-1,-1)(4,4)
    \psline(0,0)(2,0)
    \psline(0,2)(2,2)
    \psline(1,3)(3,3)
    \psline(0,0)(0,2)
    \psline(2,0)(2,2)
    \psline(3,1)(3,3)
    \psline(0,2)(1,3)
    \psline(2,0)(3,1)
    \psline(2,2)(3,3)
    \psset{linestyle=dotted}
    \psline(0,0)(1,1)
    \psline(1,1)(3,1)
    \psline(1,1)(1,3)
    \psdots(0,0)(0,2)(1,1)(1,3)(2,0)(2,2)(3,1)(3,3)
    \uput[270](0,0){$000$}
    \uput[180](0,2){$010$}
    \uput[180](1,1){$001$}
    \uput[90](1,3){$011$}
    \uput[270](2,0){$100$}
    \uput[0](2,2){$110$}
    \uput[0](3,1){$101$}
    \uput[90](3,3){$111$}
\end{pspicture}
\end{center}
\caption{A Cayley Graph for $\mathbb{Z}_{2} \oplus \mathbb{Z}_{2} \oplus \mathbb{Z}_{2}$.}\label{FigBoolCube}
\end{figure}

According to Rudeanu \cite{MR2549641}, construction of a Boolean differential calculus was first attempted by Daniell \cite{MR1559990}.  Recently, Bazs{\'o} and L{\'a}bos \cite{MR2233259} suggest that Boolean differentials should be linear and and satisfy the Leibniz identity; no proposed Boolean differential, they claim, meets both criteria.  In \cite{MR2389716} and the present paper, however, abstract differentials always satisfy the chain rule, and hence the Leibniz identity if multiplication and addition are defined on the codomain.  Thus, if we restrict the space of differentials to continuous linear functions, then we satisfy the criteria of \cite{MR2233259}. 

Since the continuous linear functions on Boolean hypercubes are precisely the continuous group homomorphisms on Boolean hypercubes, we obtain a Boolean differential calculus by specializing Lemmas \ref{LmaCntGrpHm} and \ref{LmaCnvGrpHm} to Boolean hypercubes, from which the specialization of Theorem \ref{ThmDiffGrps} follows.  Because every point of a Boolean hypercube has order 2, these specializations have simple formulations.

\begin{lemma}\label{LmaBlCntLin}
A linear function $L:B^{m} \rightarrow B^{n}$ is continuous if and only if $L(\N{0_{m}}) \subseteq \N{0_{n}}$.
\end{lemma}

\begin{lemma}\label{LmaBlLnConv}
If $L$ and $K$ are distinct linear functions in $\CC{B^{m}}{B^{n}}$, then $K \in \N{L}$ if and only if all nonzero columns of the matrix representations of $K$ and $L$ are identical.  
\end{lemma}

\begin{theorem}\label{ThmDiffBool}
If $L \in \D{B^{m}}{B^{n}}$, $f: B^{m} \rightarrow B^{n}$ is a function, and $b \in B^{m}$, then $L$ is a differential of $f$ at $b$ if and only if 
\begin{enumerate}
	\item	$L$ is isolated and $L(x) = f(x)$ for every $x \in \N{b}$;
	\item	$L$ is not isolated but constant, $f(\N{b}) \subseteq \N{0_{n}}$, and $f(0_{m}) = 0_{n}$ if $0_{m} \in \N{b}$; or
	\item	$L$ is not isolated and not constant, there exists a unique $\beta \in \N{0_{n}} - \{0_{n}\}$ such that $f(\N{b}) \subseteq L(B^{m}) = \{0_{n},\beta\}$, and $f(0_{m}) = 0_{n}$ if $0_{m} \in \N{b}$.
\end{enumerate}
\end{theorem}

A corollary to Theorem \ref{ThmDiffBool} is that every function $f:B^{m} \rightarrow B$ is differentiable at every point not in $\N{0_{m}}$ and at every point in $\N{0_{m}}$ if $f(0_{m}) = 0$.

As the next example illustrates, Boolean differential calculus is easily mechanized since it is largely a matter of solving matrix equations.

\begin{example}\label{ExmDiffBool}
Define $f:B^{2} \rightarrow B^{3}$ by $f(p,q) = (p,(1 + p)(1 + q),q)$.  Since
\begin{equation*}
	\left(	\begin{array}{cc}
		1 & 0 \\
		0 & 0 \\
		0 & 1 
		\end{array}
	\right)
	\left(	\begin{array}{ccc}
		1 & 1 & 0 \\
		1 & 0 & 1 
		\end{array}
	\right)
	=
	\left(	\begin{array}{ccc}
		1 & 1 & 0 \\
		0 & 0 & 0 \\
		1 & 0 & 1
		\end{array}
	\right)
	=
	\left(	\begin{array}{ccc}
		f(1,1) & f(1,0) & f(0,1) 
		\end{array}
	\right),
\end{equation*}
it follows that $(p,q) \mapsto (p,0,q)$ is the differential of $f$ at $(1,1)$.  Likewise, the function $g:B^{3} \rightarrow B^{2}$ defined by $g(p,q,r) = ((1 + q)(1 + p + pr),(1+r)q)$ is differentiable at $(1,0,1)$; one of its differentials at $(1,0,1)$ is $(p,q,r) \mapsto (q + r,0)$.  By Theorem \ref{ThmChainRul}, it follows that $(g \circ f)(p,q) = (q,(1+p)(1+q))$ is differentiable at $(1,1)$; one of its differentials at $(1,1)$ is $(p,q) \mapsto (q,0)$.
\end{example}

\begin{example}\label{ExmBadBlean}
Consider the function $f:B^{3} \rightarrow B^{3}$ defined by 
\begin{equation*}
	f(p,q,r) = (p(1+q)(1+r),pr(1+q),r(1+p)(1+q)).
\end{equation*}
Since $f(\N{(1,0,1)}) = \N{0_{3}}$ and $f(0_{3}) = 0_{3}$, it follows that $f$ is differentiable at $(1,0,1)$.  This function, however, is \emph{not} continuous at $(1,0,1)$ since $f(\N{(1,0,1)}) \not \subseteq \N{f(1,0,1)}$.
\end{example}

\section{Conclusion}

After reviewing preliminary results about the convergence space representation of reflexive digraphs,  we modified the definition of \emph{differential} given in \cite{MR2389716} and verified that this modification is a conservative extension of differentiation, as it is in classical analysis, to arbitrary convergence spaces.  To obtain differential calculi on groups, we identified groups with non-redundant Cayley graphs, which can be represented by convergence spaces.  By means of Boolean hypercubes---that is, direct sums of $\mathbb{Z}_{2}$---we then obtained a Boolean differential calculus, in which differentials satisfy both linearity and the Leibniz identity.  

These results are part of our larger project to develop analysis on convergence spaces and, in particular, discrete structures.  One particular focus of this project involves the development of methods to discretize continuous structures.  The theory put forth in the present paper, along with certain results of \cite{dP14}, is a foundation for analysis on discretized continuous structures.  We intimated above an approach exploiting the property that the family of all convergence structures with the same carrier forms a complete lattice if ordered by set inclusion.  Although this approach has promise, it makes no use of the present work of analysis on groups.  An alternative approach, which does use the present work, involves tiling Euclidean space (for example, the plane) with regular polygons; this method, which has parallels with recent work (for example, in Miller, et al., \cite{wM14}) in general relativity, merits exploration.

\bibliographystyle{plain}
\bibliography{DiffGrps}
\end{document}